\documentclass[10pt]{amsart}
\usepackage{epsfig,url}
\usepackage{mathdots}
\usepackage{color}

\newtheorem{theorem}{Theorem}[section]

\theoremstyle{definition}

\newtheorem{note}[theorem]{Note}

\theoremstyle{remark}

\begin{document}

\title[Evaluation of the second virial coefficient for the Mie potential]{Evaluation of the second virial coefficient for the Mie potential using the method of brackets}

\author[Iv\'{a}n Gonz\'{a}lez et al]{Ivan Gonzalez}
\address{Instituto de {F}\'{i}sicas y {A}stronom\'{i}a,~{U}niversidad de {V}alpara\'{i}so,~{A}venida~{G}ran~{B}reta\~{n}{a} $1111$,~Valpara\'{i}so,~{C}hile}
\email{ivan.gonzalez@uv.cl}

\author[]{Igor Kondrashuk}
\address{Grupo de Matem\'atica Aplicada {\rm \&} Grupo de F\'isica de Altas Energ\'ias  {\rm \&}  Centro de Ciencias Exactas   {\rm \&}  Departamento de Ciencias B\'asicas, Universidad del B\'io-B\'io, 
         Campus Fernando May, Av. Andres Bello 720, Casilla 447, Chill\'an, Chile}
\email{igor.kondrashuk@gmail.com}

\author[]{Victor H. Moll}
\address{Department of Mathematics,
Tulane University, New Orleans, LA 70118}
\email{vhm@math.tulane.edu}

\author[]{Daniel Salinas-Arizmendi}
\address{Departamento de {F}\'{i}sica,~{U}niversidad~{T}\'{e}cnica~{F}ederico~{S}anta~{M}ar\'{i}a y {C}entro {C}ient\'{i}fico 
{T}ecnol\'{o}gico de {V}alpara\'{i}so,~{C}hile ({CCTV}al.),~{C}asilla $110-{V}$,~{V}alpara\'{i}so,~{C}hile}
\email{daniel.salinas@usm.cl}

\subjclass[2010]{Primary 33}

\date{\today}


\maketitle

\newcommand{\ba}{\begin{eqnarray}}
\newcommand{\ea}{\end{eqnarray}}
\newcommand{\ift}{\int_{0}^{\infty}}
\newcommand{\nn}{\nonumber}
\newcommand{\no}{\noindent}
\newcommand{\lf}{\left\lfloor}
\newcommand{\rf}{\right\rfloor}
\newcommand{\realpart}{\mathop{\rm Re}\nolimits}
\newcommand{\imagpart}{\mathop{\rm Im}\nolimits}

\newcommand{\op}[1]{\ensuremath{\operatorname{#1}}}
\newcommand{\pFq}[5]{\ensuremath{{}_{#1}F_{#2} \left( \genfrac{}{}{0pt}{}{#3}
{#4} \bigg| {#5} \right)}}

\newmuskip\pFqmuskip

\newtheorem{Definition}{\bf Definition}[section]
\newtheorem{Thm}[Definition]{\bf Theorem}
\newtheorem{Example}[Definition]{\bf Example}
\newtheorem{Lem}[Definition]{\bf Lemma}
\newtheorem{Cor}[Definition]{\bf Corollary}
\newtheorem{Prop}[Definition]{\bf Proposition}
\numberwithin{equation}{section}

\begin{abstract}
The second virial coefficient for the Mie potential is evaluated using the method of brackets. This method converts a definite integral into a 
series in the parameters of the problem, in this case this is the temperature $T$. The results obtained here are consistent with some 
known  special cases,  such as the Lenard-Jones potential. The asymptotic properties of the second virial coefficient in molecular thermodynamic systems and complex fluid modeling are described in the limiting cases of $T \rightarrow 0$ and $T \rightarrow \infty$.
\end{abstract}


\section{Introduction}
\label{sec-introduction}

The classical virial expansion expresses the pressure $P$ of a many-particle system in equilibrium as a 
power series in the density $\rho$:
\begin{equation}
Z := \frac{P}{RT\rho} = A + B \rho +C \rho^{2} + \cdots 
\end{equation}
\noindent
where $T$ is the temperature and $Z$ is the so-called compressibility factor.  This is a 
dimensionless term which measures how much a real fluid deviates from an 
ideal gas.  The first virial coefficient $A$ is normalized to $1$ expresses the fact that, at low density, all fluids behave like ideal gases. 

The second virial coefficient is given by 
\begin{equation}
B=B(T) = - 2 \pi \int_{0}^{\infty} \left[ e^{-u(r)/kT} - 1 \right] r^{2} \, dr,
\label{form-B1}
\end{equation}
\noindent
where $u(r)$ is the intermolecular potential of the particles in the system.  This classical expression was derived by L.~Ornstein in his 
$1908$ Ph.~D. thesis. Here we analyze the case of the Mie potential 
\cite{mie-1903a}, a generalization of 
the Lennard-Jones potential given by 
\begin{equation}
u(r) = \varepsilon A \left[ \left( \frac{\sigma}{r} \right)^{n} -  \left( \frac{\sigma}{r} \right)^{m} \right],
\end{equation}
\noindent
depending on two parameters $n, \, m$ satisfying $n>m>3$ and a prefactor $A$ defined by 
\begin{equation}
A = \left( \frac{n}{n-m} \right) \left( \frac{n}{m} \right)^{  \frac{m}{n-m}}.
\end{equation}
\noindent
The Lenard-Jones potential corresponds to the values $n=12$ and $m=6$. 
The parameter $\sigma$ is related to the size of the particles, the parameters $n, \, m$ characterize the shape of the potential: $n$ for 
repulsion and $m$ for attraction,  $r$ is the relative distance among particles and $\varepsilon$ is the depth of the potential well.
The second virial coefficient is given by
\begin{equation}
\label{form-B2}
B(T_{*})  = - 2 \pi \int_{0}^{\infty} 
\left[ 
\exp \left( - \frac{1}{T_{*}} 
\left[
 \left( \frac{\sigma}{r} \right)^{n} - \left( \frac{\sigma}{r} \right)^{m}
 \right] \right)  - 1 
\right] \, r^{2} dr,
\end{equation}
\noindent 
using the notation
\begin{equation}
\label{Tstar}
\frac{1}{T_{*}} = \frac{A \varepsilon}{kT}.
\end{equation}

The goal of this note is to use the method of brackets, developed in \cite{gonzalez-2010a}, to obtain an analytic expression for the 
second virial coefficient in \eqref{form-B2}.

\section{The evaluation by classical methods}
\label{sec-evaluation-analysis}

In this section $B(T_{*})$ is evaluated by  traditional methods of mathematical analysis. An alternative approach has been presented in 
\cite{heyes-2016a}. The parameters must 
satisfy  $n>m > 3$ in order to guarantee the convergence 
of the integral appearing in this expression.  An artificial parameter $\Lambda$ is introduced in order to deal with convergence issues in 
\eqref{form-B2}. In addition,  a new parameter $\epsilon$ is introduced and is used to treat the first term in the resulting series. Then 
\eqref{form-B2} is written as 
\begin{eqnarray} \label{limits} 
B(T_{*}) & = &  - 2 \pi  \lim\limits_{\epsilon \rightarrow 0}\int_{0}^{\infty} \left[ \exp \left( - \frac{1}{T_{*}} \left[\left( \frac{\sigma}{r} \right)^{n} - \left( \frac{\sigma}{r} \right)^{m}\right] \right) -1 \right]  r^{2-n\epsilon} dr \\
& = &  - 2 \pi  \lim\limits_{\epsilon \rightarrow 0}\lim\limits_{\Lambda \rightarrow \infty} \int_{0}^{\infty} 
\left[ \exp \left( - \frac{1}{T_{*}} \left[\left( \frac{\sigma}{r} \right)^{n} - \left( \frac{\sigma}{r} \right)^{m}\right] \right) \right. \nonumber \\ 
& &\quad \quad  - \left.  \exp \left( - \frac{1}{\Lambda} \left[\left( \frac{\sigma}{r} \right)^{n} - \left( \frac{\sigma}{r} \right)^{m}\right] \right)\right]  r^{2-n\epsilon} dr.
\nonumber
\end{eqnarray}

The parameter $\epsilon$ is kept finite and large enough to regularize the Euler gamma function appearing 
in the first term of the resulting series. Then, some simple transformations  lead to 
\begin{multline*} 
\int_{0}^{\infty} \left[ \exp \left( - \frac{1}{T_{*}} \left[\left( \frac{\sigma}{r} \right)^{n} - \left( \frac{\sigma}{r} \right)^{m}\right] \right)  \right]
  r^{2-n\epsilon} dr 
  \end{multline*}
  \begin{eqnarray*}
   & = & 
 \int_{0}^{\infty} \left[ \exp \left( - \frac{1}{T_{*}} \left[\left( {\sigma}{r} \right)^{n} - \left( {\sigma}{r} \right)^{m}\right] \right)  \right]  r^{n\epsilon -4} dr \\
& = &  {\sigma}^{3-n\epsilon}\int_{0}^{\infty} \left[ \exp \left( - \frac{1}{T_{*}} \left[r^n - r^m\right] \right)  \right]  r^{n\epsilon -4} dr \\
& = &   {\sigma}^{3-n\epsilon} \sum_{k=0}^{\infty} \frac{1}{k!T_{*}^k}  \int_{0}^{\infty} r^{mk +n\epsilon -4 }   \exp \left( - \frac{r^n}{T_{*}} \right)    dr \\
& = & \frac{{\sigma}^{3-n\epsilon}}{n} \sum_{k=0}^{\infty} \frac{1}{k!T_{*}^k}  \int_{0}^{\infty} r^{(mk +n\epsilon -3 )/n -1}   \exp \left( - \frac{r}{T_{*}} \right)    dr \\ 
& = &  \frac{{\sigma}^{3-n\epsilon}}{n}T_{*}^{\epsilon - 3/n} \sum_{k=0}^{\infty} \frac{T_{*}^{mk/n}}{k!T_{*}^k}  \int_{0}^{\infty} r^{(mk +n\epsilon -3 )/n -1}   e^{ - r }    dr \\
& = & \frac{{\sigma}^{3-n\epsilon}}{n}T_{*}^{\epsilon - 3/n} \sum_{k=0}^{\infty} \frac{T_{*}^{(m-n)k/n}}{k!}  \Gamma \left( \frac{mk-3}{n} + \epsilon \right).
\end{eqnarray*}
\noindent
Observe that $\epsilon > 3/n$ is required to evaluate the integral in the last line in terms of the Gamma function. The final identity 
\begin{multline}
\int_{0}^{\infty} \left[ \exp \left( - \frac{1}{T_{*}} \left[\left( \frac{\sigma}{r} \right)^{n} - \left( \frac{\sigma}{r} \right)^{m}\right] \right)  \right]
  r^{2-n\epsilon} dr =  \\ \frac{{\sigma}^{3-n\epsilon}}{n}T_{*}^{\epsilon - 3/n} \sum_{k=0}^{\infty} \frac{T_{*}^{(m-n)k/n}}{k!}  \Gamma \left( \frac{mk-3}{n} + \epsilon \right)
  \end{multline}
  \noindent
  is analytic in $\epsilon$, so we might let $\epsilon \rightarrow 0$ and $\Lambda \rightarrow \infty$ to produce 
%
\begin{multline*}
B(T_{*}) = -   \frac{2\pi}{n} \lim\limits_{\epsilon \rightarrow 0} {\sigma}^{3-n\epsilon}T_{*}^{\epsilon - 3/n} \sum_{k=0}^{\infty} \frac{T_{*}^{(m-n)k/n}}{k!}  \Gamma \left( \frac{mk-3}{n} + \epsilon \right) \\
 -    \frac{2\pi}{n} \lim\limits_{\epsilon \rightarrow 0}\lim\limits_{\Lambda \rightarrow \infty} {\sigma}^{3-n\epsilon}\Lambda^{\epsilon - 3/n} \sum_{k=0}^{\infty} \frac{\Lambda^{(m-n)k/n}}{k!}  \Gamma \left( \frac{mk-3}{n} + \epsilon \right).  
\end{multline*}
The second term vanishes in the limit $\Lambda \rightarrow \infty$ and when  $\epsilon \rightarrow 0$ 
it gives the known result 
\begin{eqnarray*}
B(T_{*}) & = &  - \frac{2\pi}{n}\lim\limits_{\epsilon \rightarrow 0} {\sigma}^{3-n\epsilon}T_{*}^{\epsilon - 3/n} \sum_{k=0}^{\infty} \frac{T_{*}^{(m-n)k/n}}{k!}  \Gamma \left( \frac{mk-3}{n} + \epsilon \right) \\
& = &  - \frac{2\pi}{n} {\sigma}^3 T_{*}^{- 3/n} \sum_{k=0}^{\infty} \frac{T_{*}^{(m-n)k/n}}{k!}  \Gamma \left( \frac{mk-3}{n}\right).
\end{eqnarray*}
\noindent 
Observe that the term for $k=0$ in the expression for $B(T_{*})$, as $\epsilon \rightarrow 0$,  contains the term 
\begin{equation*}
\lim\limits_{\varepsilon \rightarrow 0} \Gamma \left( - \frac{3}{n} + \varepsilon \right) = 
\lim\limits_{\varepsilon \rightarrow 0} \frac{\Gamma(-3/n + \varepsilon + 1)}{-3/n + \varepsilon}
 = \frac{\Gamma(-3/n+1)}{(-3/n)} = \Gamma \left( - \frac{3}{n} \right).
 \end{equation*}
 \noindent
 and this  limiting value is finite since $n>3$. Therefore  letting $\varepsilon \rightarrow 0$ does not produce singularities.

\section{The method of brackets}
\label{sec-brackets}

Section \ref{sec-evaluation} presents the evaluation of the second virial coefficient 
$B(T_{*})$ using the method of brackets. This is a method of integration, based on 
a small number of rules described here. A complete description of this method as well as a variety of definite integrals evaluated using 
it may be found in 
\cite{amdeberhan-2012b,bravo-2017a,gonzalez-2010c,gonzalez-2010a,gonzalez-2014a,gonzalez-2017a,gonzalez-2017b,gonzalez-2020a,gonzalez-2022a}.

The method of brackets evaluates an integral of the form 
\begin{equation}
I = \int_{0}^{\infty} f(x) \, dx
\end{equation}
\noindent
where the function $f$ has an expansion of the form 
\begin{equation}
f(x)  = \sum_{n=0}^{\infty} C(n) x^{\alpha n + \beta -1}
\end{equation}
\noindent
with $\alpha, \, \beta \in \mathbb{C}$. (The extra $-1$ in the exponent is just a convenience for future formulas).

The basic concept is the definition of the \texttt{bracket} by the integral 
\begin{equation}
\langle  b \rangle = \int_{0}^{\infty} x^{b-1} \, dx
\end{equation}
\noindent
and (by linearity) this gives 
\begin{equation}
I = \int_{0}^{\infty} f(x) \, dx = \int_{0}^{\infty} \sum_{n=0}^{\infty} C(n) x^{\alpha n + \beta -1} = \sum_{n=0}^{\infty} C(n)
\langle \alpha n + \beta \rangle.
\end{equation}
\noindent
The expression on the right is called a \texttt{bracket series}.  The method consists of a sequence of rules to generate and evaluate
such series. It is convenient to introduce the so-called \texttt{indicator} defined by 
\begin{equation}
\phi_{n} = \frac{(-1)^{n}}{n!} = \frac{(-1)^{n}}{\Gamma(n+1)}.
\end{equation}

\smallskip 

\noindent
\texttt{Rule 1}.  To an integral of the form 
\begin{equation*}
\int_{0}^{\infty} \sum_{n} \phi_{n} C(n) x^{\alpha n + \beta - 1} \, dx
\end{equation*}
\noindent
one assigns the bracket series 
\begin{equation*}
\sum_{n} \phi_{n} C(n) \langle \alpha n + \beta \rangle.
\end{equation*}
\noindent
This bracket series is assigned the value 
\begin{equation*}
\frac{1}{| \alpha | } C(n^{*}) \Gamma(-n^{*})
\end{equation*}
\noindent
where $n^{*}$ is the unique solution to  $\alpha n + \beta  = 0$. Observe that this requires an extension of the function $C$ defined originally 
for indices $n \in \mathbb{N}$ to $\mathbb{C}$. 

\smallskip

\noindent
\texttt{Rule 2}. Let $A = (\alpha_{ij})$ be a nonsingular matrix. The multidimensional extension of Rule 1 is as follows: To an integral of 
the form 
\begin{equation*}
\int_{0}^{\infty} \cdots \int_{0}^{\infty} \sum_{n_{1}, \cdots,  n_{k}} C(n_{1}, \cdots, n_{k}) 
x_{1}^{\alpha_{11}n_{1} + \cdots + \alpha_{1k}n_{k} + \beta_{1} - 1} \cdots x_{k}^{\alpha_{k1}n_{1} + \cdots + \alpha_{kk}n_{k} + \beta_{k} - 1} 
dx_{1} \cdots dx_{k}
\end{equation*}
\noindent
one assigns the multidimensional bracket series 
\begin{equation*}
\sum_{n_{1}, \cdots, n_{k}} \phi_{n_{1} \cdots n_{k}} C(n_{1}, \cdots, n_{k}) 
\langle \alpha_{11} n_{1} + \cdots + \alpha_{1k}n_{k} + \beta_{1} \rangle \cdots 
\langle \alpha_{k1} n_{1} + \cdots + \alpha_{kk}n_{k} + \beta_{k} \rangle,
\end{equation*}
\noindent
with $\phi_{n_{1} \cdots n_{k}}= \phi_{n_{1}} \cdots \phi_{n_{k}}$. To this bracket series, one assigns the value 
\begin{equation}
\frac{1}{| \det(A) |} C(n_{1}^{*}, \cdots, n_{k}^{*}) \prod_{j=1}^{k} \Gamma(-n_{j}^{*})
\end{equation}
\noindent
where $n_{1}^{*}, \ldots, n_{k}^{*}$ is the unique solution to the linear system given by the vanishing of the brackets.  If the matrix $A$ is 
singular, the method is inconclusive.  The issue of extending $C$ is treated as in the one-dimensional case. 

\smallskip

\noindent
\texttt{Rule 3}. This deals with the situation of a multidimensional bracket series in which the number of brackets is fewer than the 
number of indices in the sum. Then one must choose free indices from $n_{1}, \cdots, n_{k}$, equal in number to that of the brackets appearing. 
For each of these choices a series in the free indices, called a \texttt{basis series}, is obtained by applying Rule 2. If the basis series is 
divergent, then it is discarded. The value of the integral is obtained by summing the basis series which converge in a common region. In 
general, there will be multiple series solutions built from the basis series, each one of which is a series representation for the integral in the 
common region of convergence. 


\section{The evaluation by the method of brackets}
\label{sec-evaluation}

This section uses the method of brackets to evaluate  $B(T_{*})$. The expression is given as a series in $T_{*}$ and 
the parameters must satisfy  $n>m > 3$ in order to guarantee its  convergence.
Section \ref{sec-specialcases} shows how to evaluate this series in the case  when the quotient $m/n$ is a rational number.  
A special function central to these 
evaluation is defined next. 

\begin{Definition}
\label{def-kummer}
The hypergeometric function $_{1}F_{1}$, given by 
\begin{equation}
\pFq11{a}{b}{z} = \sum_{k=0}^{\infty} \frac{(a)_{k}}{(b)_{k}} \frac{z^{k}}{k!}
\end{equation}
\noindent
is referred in the literature as the \texttt{Kummer function}.
\end{Definition}

The computation of  $B(T_{*})$ is divided into two parts. In order to take into account the term $-1$ in the second integral, an artificial parameter 
$\Lambda$ is introduced. In the limit $\Lambda \rightarrow \infty$ the second integral vanishes. The role of the 
parameter $\Lambda$ is to guarantee the convergence of the expression for $B(T_{*})$ given below. Define 
\begin{equation}
\label{form-B2a}
J_{T_{*}}  = - 2 \pi \int_{0}^{\infty} 
\left[ 
\exp \left( - \frac{1}{T_{*}} 
\left[
 \left( \frac{\sigma}{r} \right)^{n} - \left( \frac{\sigma}{r} \right)^{m}
 \right] \right)  
\right] \, r^{2} dr,
\end{equation}
\noindent
and 
\begin{equation}
\label{form-B2b}
J_{\Lambda} = - 2 \pi \int_{0}^{\infty} 
\left[ 
\exp \left( - \frac{1}{\Lambda}
\left[
 \left( \frac{\sigma}{r} \right)^{n} - \left( \frac{\sigma}{r} \right)^{m}
 \right] \right)  
\right] \, r^{2} dr.
\end{equation}
\noindent
Then 
\begin{equation}
\label{B-Tstar}
B(T_{*}) = J_{T_{*}} - \lim\limits_{\Lambda \rightarrow \infty} J_{\Lambda}.
\end{equation}

\noindent
\texttt{Step 1}. The computation starts with producing a bracket series for the integral $J_{T_{*}}$. This comes directly from the expansion of the 
exponential function:
\begin{eqnarray}
& & \\ 
J_{T_{*}} & = & -2 \pi \int_{0}^{\infty} \exp \left( - \frac{1}{T_{*}} \left( \frac{\sigma}{r} \right)^{n} \right) 
\exp \left( \frac{1}{T_{*}} \left( \frac{\sigma}{r} \right)^{m} \right) r^{2} \, dr  \nonumber \\
& = & - 2 \pi \int_{0}^{\infty} 
\left[ \sum_{\ell \geq 0} \phi_{\ell} \left( \frac{1}{T_{*}} \right)^{\ell} \sigma^{\ell n} r^{- \ell n} \right] 
\left[ \sum_{j \geq 0} \phi_{j} (-1)^{j} \left( \frac{1}{T^{*}} \right)^{j} \sigma^{j m} r^{-j m} \right] r^{2} \, dr \nonumber \\
& = & - 2 \pi \sum_{\ell \geq 0 } \sum_{j \geq 0} \phi_{\ell j} (-1)^{j} \left( \frac{1}{T_{*}} \right)^{\ell + j } \sigma^{\ell n + m j } 
\langle 3 - \ell n - m j \rangle. \nonumber
\end{eqnarray}

\smallskip

\noindent
\texttt{Step 2}. The evaluation of the bracket series in Step 1 produces two series, one per free index $\ell$ or $j$. A direct 
computation gives 
\begin{equation}
J_{T_{*}}^{(1)} = (-1)^{1+\tfrac{3}{m}} \frac{2 \pi \sigma^{3}}{m T_{*}^{3/m}}
\sum_{k=0}^{\infty} \frac{1}{k!} \Gamma \left( \frac{nk-3}{m} \right) (-1)^{(m-n)k/m} T_{*}^{(n-m)k/m},
\end{equation}
\noindent
and 
\begin{equation}
J_{T_{*}}^{(2)} = - \frac{2 \pi \sigma^{3}}{n T_{*}^{3/n}}\sum_{k=0}^{\infty} \frac{1}{k!} 
\Gamma \left( \frac{km-3}{n} \right) \left( \frac{1}{T_{*}} \right)^{(n-m)k/n}. 
\end{equation}
\noindent
Then 
\begin{equation}
J_{T_{*}} =
J_{T_{*}}^{(1)} 
\quad \textnormal{or}  \quad 
J_{T_{*}}^{(2)},
\end{equation}
since these two expressions are expansions in the distinct arguments $T_{*}$ or $1/T_{*}$.

The appearance of the term $(-1)^{(m-n)k/m}$ in the expansion of $J_{T_{*}}^{(1)}$ shows that this series is not real, since 
$n>m>3$. Therefore it is 
discarded. It follows that $J_{T_{*}}^{(2)}$ is the only admissible solution.  This is written as 
\begin{equation}
\label{new-series}
J_{T_{*}}= - \frac{2 \pi \sigma^{3}}{n}  T_{*}^{-3/n} \sum_{k=0}^{\infty} \frac{1}{k!} 
\Gamma \left( \frac{km-3}{n} \right) T_{*}^{-(n-m)k/n},
\end{equation}
and observe that all the exponents of $T_{*}$ in \eqref{new-series} are negative. 

The second term in \eqref{B-Tstar} is obtained by replacing $T_{*}$ by $\Lambda$ in \eqref{new-series}.
A  direct calculation shows  that 
 $\lim\limits_{\Lambda \rightarrow \infty} J_{\Lambda} = 0$. Therefore this integral does not contribute at $\Lambda = \infty$.  This proves:

\begin{theorem}
\label{virial-thm1}
The second virial coefficient $B(T_{*})$ is given by the series 
\begin{equation}
B(T_{*})= - \frac{2 \pi \sigma^{3}}{n}  T_{*}^{-3/n} \sum_{k=0}^{\infty} \frac{1}{k!} 
\Gamma \left( \frac{km-3}{n} \right) T_{*}^{-(n-m)k/n}, 
\end{equation}
\noindent
where $n > m > 3$.
\end{theorem}

\section{Special cases}
\label{sec-specialcases}

The expression for $B(T_{*})$ in Theorem \ref{virial-thm1} is now denoted by $B(n,m,\sigma,T_{*})$ to include the dependence on all 
its  parameters.  It turns out that $B(T_{*})$ simplifies when the ratio $n/m$ is a rational number. For instance, a \texttt{Mathematica} 
evaluation gives 
\begin{eqnarray}
B(8,4,\sigma,T_{*}) & = &  -\frac{\pi  \sigma^3 \Gamma \left(-\frac{3}{8}\right) }{4 T_{*}^{3/8}}  \, 
_1F_1\left(-\frac{3}{8};\frac{1}{2};\frac{1}{4 T_{*}} \right)  \\
& & \quad \quad \quad  -
\frac{\pi  \sigma^3 \Gamma \left(\frac{1}{8}\right) }{ 4 T_{*}^{11/8}} \, 
_1F_1\left(\frac{1}{8};\frac{3}{2};\frac{1}{4 T_{*}}\right) \nonumber 
\end{eqnarray}
\noindent
and 
\begin{eqnarray}
B(10,4,\sigma,T_{*}) & = & 
-\frac{\pi ^{3/2} \sigma^3 \, _2F_4\left(\frac{1}{4},\frac{3}{4};\frac{3}{5},\frac{4}{5},\frac{6}{5},\frac{7}{5};\frac{4}{3125 T_{*}^3}\right)}{10 T_{*}^{3/2}}
\\
& & -\frac{\pi  \sigma^3 \Gamma \left(-\frac{3}{10}\right) \, _2F_4\left(-\frac{3}{20},\frac{7}{20};\frac{1}{5},\frac{2}{5},\frac{3}{5},\frac{4}{5};\frac{4}{3125 T_{*}^3}\right)}{5 T_{*}^{3/10}} \nonumber \\
& & -\frac{\pi  \sigma^3 \Gamma \left(\frac{1}{10}\right) \, _2F_4\left(\frac{1}{20},\frac{11}{20};\frac{2}{5},\frac{3}{5},\frac{4}{5},\frac{6}{5};\frac{4}{3125 T_{*}^3}\right)}{5 T_{*}^{9/10}} \nonumber \\
& & -\frac{\pi  \sigma^3 \Gamma \left(\frac{9}{10}\right) \, _2F_4\left(\frac{9}{20},\frac{19}{20};\frac{4}{5},\frac{6}{5},\frac{7}{5},\frac{8}{5};\frac{4}{3125 T_{*}^3}\right)}{30 T_{*}^{21/10}}  \nonumber \\
& & -\frac{\pi  \sigma^3 \Gamma \left(\frac{3}{10}\right) \, _2F_4\left(\frac{13}{20},\frac{23}{20};\frac{6}{5},\frac{7}{5},\frac{8}{5},\frac{9}{5};\frac{4}{3125 T_{*}^3}\right)}{400 T_{*}^{27/10}}. \nonumber 
\end{eqnarray}

It is possible to show that when $m/n$ is a rational number, the expression for $B(n,m,\sigma,T_{*})$ given in 
Theorem \ref{virial-thm1} may be 
reduced to a finite sum of hypergeometric series of the form $_{p}F_{q}$ with $p \leq q$. These have infinite radius of convergence.  This 
is now discussed in detail in the special case $n=2m$.

\subsection{The second virial coefficient for the Mie potential with $n=2m$.} Separating the sum 
according to the parity of the index, the expression for $B(T_{*})$ becomes 
\begin{equation*}
B(T_{*}) = - \frac{\pi \sigma^{3}}{m T_{*}^{3/2m}} 
\left[ \sum_{k=0}^{\infty} \frac{\Gamma \left( k - \tfrac{3}{2m} \right) }{(2k)!} \left( \frac{1}{T_{*}} \right)^{k} + 
\frac{1}{\sqrt{T_{*}}}  \sum_{k=0}^{\infty} \frac{\Gamma \left( k + \tfrac{1}{2} - \tfrac{3}{2m} \right) }{(2k+1)!} \left( \frac{1}{T_{*}} \right)^{k} 
\right].
\end{equation*}


The next result rewrites the previous expression for $B(T_{*})$ in terms of the Kummer function  $_{1}F_{1}$ defined in \eqref{def-kummer}.

\medskip

\begin{theorem}
\label{thm-kummer1}
The second virial coefficient $B(T_{*})$  can be written in terms of the Kummer function $_{1}F_{1}$ in the form
\begin{equation*}
B(T_{*}) =  - \frac{\pi \sigma^{3}}{m T_{*}^{3/2m}} 
 \left[ 
 \Gamma \left( - \tfrac{3}{2m} \right) 
 \pFq11{-\tfrac{3}{2m}}{\tfrac{1}{2}}{\frac{1}{4T_{*}} }+
\frac{ \Gamma \left( \tfrac{1}{2} - \tfrac{3}{2m} \right) }
 {\sqrt{T_{*}}} 
  \pFq11{\tfrac{1}{2} -\tfrac{3}{2m}}{\tfrac{3}{2}}{\frac{1}{4T_{*}}}
  \right].
  \end{equation*}
  \end{theorem}
  
  \medskip
  
  The expression above is now examined in the limiting cases $T_{*} \rightarrow \infty$ and $T_{*} \rightarrow 0$. 
  
  \smallskip 
  
  \noindent
  \texttt{Behavior at $T_{*} \rightarrow \infty$}. This can be read directly: the asymptotic 
  \begin{equation}
  B(T_{*}) \sim -  \frac{1}{m} \pi \sigma^{3} \Gamma \left( - \frac{3}{2m} \right) T_{*}^{-3/2m}>0
  \end{equation}
  \noindent 
  follows  from the expression for $B(T_{*})$.
  
  \medskip
  
  \noindent
  \texttt{Behavior at $T_{*} \rightarrow 0$}. The series for $\pFq11{a}{b}{z}$ 
  converges for all $z \in \mathbb{C}$.  Two transformations for the Kummer function given below play a crucial role.
  Properties of Kummer function  may also be found in \cite{alvarez-2020a}.

  \begin{Lem}
  \label{convert-1}
  The Kummer function satisfies 
  \begin{equation}
  \pFq11{a}{b}{x} = e^{x} \pFq11{b-a}{b}{-x}
  \end{equation}
  \noindent
  and
  \begin{equation}
  \label{form-2}
  \pFq11{a}{b}{x} = e^{x} x^{-b+a} \frac{\Gamma(b)}{\Gamma(a)} \pFq20{1-a \,\,\, 1-b}{-}{\frac{1}{x}}  + o(1)
  \end{equation}
  \noindent
  valid as $x \rightarrow \infty$. 
  \end{Lem}
  \begin{proof}
  The first transformation appears  as entry $13.1.27$ in \cite{abramowitz-1972a}   and also as Exercise $9$ in \cite[Chapter 2]{andrews-1999a}. 
  The proof of the second formula appears in Section $4.7$ of \cite{luke-1969a}. 
   \end{proof}

  The asymptotic behavior as $T_{*} \rightarrow 0$ is now obtained from the formula in Theorem \ref{thm-kummer1}.  The 
  identity \eqref{form-2} in the limit $T_{*} \rightarrow 0$,  using the fact that $x$ is proportional to $1/T_{*}$,  is written as 
  \begin{equation}
  \label{form-2a}
  \pFq11{a}{b}{x} \sim  e^{x} x^{-b+a} \frac{\Gamma(b)}{\Gamma(a)} \pFq20{1-a \,\,\, 1-b}{-}{\frac{1}{x}}.
  \end{equation}
  \noindent
  This produces 
  \begin{equation}
  B(T_{*}) = - 2^{(2+3/m)} \frac{\pi^{3/2} \sigma^{3}}{m} \sqrt{T_{*}} \exp \left( \frac{1}{4 T_{*}} \right) 
  \pFq20{ \frac{2m+3}{2m} \,\,\,\, \frac{m+3}{2m}}{-}{4 T_{*}}
  \end{equation}
  \noindent
  and since $T_{*} = kT/4 \varepsilon$, it follows that
  \begin{equation}
  B(T) = - 2^{1+ 3/m} \frac{\pi^{3/2} \sigma^{3}}{m} \sqrt{ \frac{k T}{\varepsilon}} 
  \exp \left( \frac{\varepsilon}{kT} \right) \, \pFq20{ \frac{2m+3}{2m} \,\,\,\, \frac{m+3}{2m}}{-}{\frac{kT}{\varepsilon}}.
  \end{equation}
  in the limit as $T \rightarrow 0$.

  The series above could be truncated to an arbitrary order to obtain the required asymptotic approximation. For instance, to order $2$ in $T$ one 
  obtains 
  \begin{multline}
  B(T) = - 2^{1+m/3} \frac{\pi^{3/2} \sigma^{3}}{m} \sqrt{\frac{kT}{\varepsilon}} \exp \left( \frac{\varepsilon}{kT} \right) \\
  \times \left( 1 + \frac{(m+3)(2m+3)}{4m^{2}} \frac{kT}{\varepsilon}+ 
  \frac{3(m+3)(2m+3)(m+1)(4m+3)}{32m^{4}} \left( \frac{ kT}{\varepsilon} \right)^{2} + \mathcal{O}(T^{3}) \right). 
  \end{multline}

  \begin{Example}
  The result in Theorem \ref{thm-kummer1} in the special case of the Lennard-Jones potential $(n=12, \, m=6)$ yields
  \begin{equation}
  B(T_{*}) = - \frac{\pi \sigma^{3}}{6 T_{*}^{1/4}} 
  \left[ \Gamma \left( - \tfrac{1}{4} \right) \pFq11{-\tfrac{1}{4}}{\tfrac{1}{2}}{ \frac{1}{4 T_{*} }}  + 
 \frac{ \Gamma \left( \tfrac{1}{4} \right) }{\sqrt{T_{*}}} \pFq11{\tfrac{1}{4}}{\tfrac{3}{2}}{ \frac{1}{4 T_{*} }} \right]
  \end{equation}
  \noindent
  and using \eqref{Tstar} this may be written as 
  \begin{equation}
  \label{form-37}
  B(T) = - \frac{\pi \sigma^{3}}{3 \sqrt{2} \left( \frac{kT}{\varepsilon} \right)^{\tfrac{1}{4}}} 
  \Gamma \left( - \tfrac{1}{4} \right) \,
   \pFq11{-\tfrac{1}{4}}{\tfrac{1}{2}}{\frac{\varepsilon}{kT}} -
  \frac{\sqrt{2} \pi \sigma^{3}}{3 \left( \frac{kT}{\varepsilon} \right)^{\tfrac{3}{4}}} \Gamma \left( \tfrac{1}{4} \right) 
   \pFq11{\tfrac{1}{4}}{\tfrac{3}{2}}{\frac{\varepsilon}{kT}}.
   \end{equation}
   \noindent
   This expression appears in \cite{gonzalez-calderon-2015a}.
  \end{Example}
  
  \begin{note}
  The Kummer function may be written as a linear combination of modified Bessel functions of first kind $I_{\nu}(x)$ using the 
  identity \cite[Formula $13.6.11-1$]{olver-2010a}
  \begin{equation}
  \pFq11{\nu+ \tfrac{1}{2}}{2 \nu + 1 + n }{2z} = \Gamma(\nu) e^{z} \left( \frac{z}{2} \right)^{-\nu} 
  \sum_{k=0}^{n} \frac{(-n)_{k} (2 \nu)_{k} (\nu+k)}{(2 \nu + 1 + n)_{k}  \, k!} I_{\nu+k}(z).
  \end{equation}
  \noindent
  From here  the terms appearing in \eqref{form-37} may be  written as 
  \begin{equation}
  \pFq11{- \tfrac{1}{4}}{\tfrac{1}{2}}{\frac{\varepsilon}{kT}} = 
  \frac{\pi}{2 \Gamma \left( \tfrac{3}{4} \right)} \left( \frac{\varepsilon}{kT} \right)^{3/4} 
  \exp \left( \frac{\varepsilon}{2kT} \right)
  \left[ I_{-\tfrac{3}{4}} \left( \frac{\varepsilon}{2 kT} \right) - I_{\tfrac{1}{4}} \left( \frac{\varepsilon}{2 kT} \right) \right],
  \end{equation}
  \noindent
  and 
  \begin{equation}
  \pFq11{ \tfrac{1}{4}}{\tfrac{3}{2}}{\frac{\varepsilon}{kT}} = 
  \frac{ \Gamma \left( \tfrac{3}{4} \right)}{\sqrt{2}}  \left( \frac{\varepsilon}{kT} \right)^{1/4} 
  \exp \left( \frac{\varepsilon}{2kT} \right) 
  \left[ I_{-\tfrac{1}{4}} \left( \frac{\varepsilon}{2 kT} \right) - I_{\tfrac{3}{4}} \left( \frac{\varepsilon}{2 kT} \right) \right].
  \end{equation}
  \noindent
  Replacing in \eqref{form-37} leads to 
  \begin{multline}
  B(T) = \frac{\pi^{2} \sigma^{3} }{3} \, \frac{\varepsilon}{kT} \exp \left( \frac{\varepsilon}{2kT} \right) \times \\
  \left[ 
  I_{-\tfrac{3}{4}} \left( \frac{\varepsilon}{2kT} \right) 
  + I_{\tfrac{3}{4}} \left( \frac{ \varepsilon}{2kT} \right)  
  - I_{\tfrac{1}{4}} \left( \frac{\varepsilon}{2kT} \right)  
  - I_{-\tfrac{1}{4}} \left(\frac{ \varepsilon}{2kT} \right)   
  \right].
  \end{multline}
  \noindent
  This result agrees with the one established in \cite{vargas-2001a}. 
  \end{note}
  
  \begin{note}
The  expression \eqref{form-37} and \eqref{form-2} give the value 
\begin{equation}
B(T) = - \frac{\sqrt{2} \pi^{3/2} \sigma^{3}}{3} \sqrt{ \frac{k T}{\varepsilon}} 
\exp \left( \frac{\varepsilon}{kT} \right) \pFq20{ \tfrac{3}{2} \quad \tfrac{5}{4}}{-}{ \frac{kT}{\varepsilon}}
\end{equation}
\noindent 
or equivalently 
\begin{eqnarray*}
B(T) & = & - \frac{\sqrt{2} \pi^{3/2} \sigma^{3}}{3} \sqrt{ \frac{kT}{\varepsilon}}
\exp \left( \frac{\varepsilon}{kT} \right) 
\sum_{n=0}^{\infty} \frac{\Gamma(n+ \tfrac{3}{4}) \Gamma(n + \tfrac{5}{4})}{\Gamma( \tfrac{3}{4}) \Gamma(\tfrac{5}{4})} 
\frac{ \left( \tfrac{kT}{\varepsilon} \right)^{n}}{n!} \\
& = & - \frac{\sqrt{2} \pi^{3/2} \sigma^{3}}{3} \sqrt{ \frac{kT}{\varepsilon}} \exp \left( \frac{\varepsilon}{kT} \right) 
\left[ 1 + \frac{15}{16} \left( \frac{kT}{\varepsilon} \right) + 
\frac{945}{512} \left( \frac{kT}{\varepsilon} \right)^{2} + \frac{45045}{8192} \left( \frac{kT}{\varepsilon} \right)^{3} + \cdots \right], \nonumber 
\end{eqnarray*}
\noindent
a result appearing  in \cite{gonzalez-calderon-2015a}.
\end{note}

 \subsection{A second example for the  Mie potential: $\mathbf{n=9}$ and $\mathbf{m=6}$.} In this case, the expression for $B(T_{*})$ in Theorem 
 \ref{virial-thm1} gives 
 \begin{equation}
 B(T_{*}) = - \frac{2 \pi \sigma^{3}}{9 T_{*}^{1/3}} \sum_{k=0}^{\infty} \frac{1}{k!} \Gamma \left( \frac{2k-1}{3} \right) 
 \frac{1}{T_{*}^{k/3}}.
 \end{equation}
 \noindent
 The index $k$ is now separated into the three classes modulo $3$. This gives a hypergeometric 
 representation of the second virial coefficient:
 \begin{eqnarray}
 \label{form-46} && \\
 B(T_{*}) & = &- \frac{ 2 \pi \sigma^{3}}{9 T_{*}^{1/3}} \Gamma \left( -\tfrac{1}{3} \right) 
 \pFq11{-\tfrac{1}{6}}{\tfrac{2}{3}}{ \,\, \frac{4}{27 T_{*}}}  \nonumber \\
 & & \quad -  \frac{ 2 \pi \sigma^{3}}{9 T_{*}^{2/3}} \Gamma \left( \tfrac{1}{3} \right) 
 \pFq11{\tfrac{1}{6}}{\tfrac{4}{3}}{ \,\, \frac{4}{27 T_{*}}}
  -  \frac{  \pi \sigma^{3}}{9 T_{*}}
 \pFq22{\tfrac{1}{2} \,\,\,\, 1 }{\tfrac{4}{3} \,\,\,\, \tfrac{5}{3}}{\,\,\frac{4}{27 T_{*}}}.
  \nonumber 
 \end{eqnarray}
 
 The expression above is useful to compute the limiting behavior as $T_{*} \rightarrow \infty$. To determine the behavior as 
 $T_{*} \rightarrow 0$, use the transformation in Lemma \ref{convert-1} to transform the $_{1}F_{1}$ into $_{2}F_{0}$ and the 
 relation 
 \begin{equation}
 \pFq22{ 1 \quad \tfrac{1}{2}}{\tfrac{4}{3} \quad \tfrac{5}{3}}{ \frac{4}{27T_{*}}} = 
 \frac{9 \sqrt{\pi}}{2} T_{*}^{3/2} \exp \left( \frac{4}{27 T_{*}} \right)
 \pFq20{\tfrac{5}{6} \quad \tfrac{7}{6} }{-}{\frac{27 T_{*}}{4}}
 \end{equation}
 \noindent
given in \cite{kim-1972a},  to obtain the hypergeometric representation
\begin{equation}
B(T_{*}) = - \frac{3 \pi^{3/2}}{2} \sigma^{3} \sqrt{T_{*}} \exp \left( \frac{4}{27 T_{*}} \right) 
\pFq20{\tfrac{5}{6} \quad \tfrac{7}{6} }{-}{\frac{27 T_{*} }{4}}
\end{equation}
\noindent
with 
\begin{equation}
\frac{1}{T_{*}} =
 \frac{27 \varepsilon}{4kT}.
\end{equation}
\noindent
Therefore, as $T_{*}  \rightarrow 0$, it follows that 
\begin{equation}
B(T) = - \frac{\pi^{3/2} \sigma^{3}}{\sqrt{3}} \sqrt{ \frac{kT}{\varepsilon}} \exp \left( \frac{ \varepsilon}{kT} \right) 
\pFq20{\tfrac{5}{6}  \quad   \tfrac{7}{6}  } {-}{ \frac{kT}{\varepsilon}}.
\end{equation}
\noindent
This can be used to obtain an asymptotic expansion to any order. For instance, up to order $2$, 
\begin{equation}
B(T) = - \frac{\pi^{3/2} \sigma^{3}}{\sqrt{3}} \sqrt{ \frac{kT}{\varepsilon}} \exp \left( \frac{\varepsilon}{kT} \right) 
\left( 1 + \frac{35}{36} \frac{kT}{\varepsilon} + 
\frac{5005}{2592} \left( \frac{kT}{\varepsilon} \right)^{2} + \mathcal{O}(T^{3}) \right).
\end{equation}

\section{Conclusions}
\label{sec-conclusions}

The virial coefficients appear in the expansion of pressure of a many-particle system as a power series in the density. The second  
virial coefficient $B(T)$, has a definite integral expression in terms of the intermolecular potential.  In the case of the Mie potential, a 
generalization of the classical Lennard-Jones potential, we have evaluated this integral by the method of brackets and obtained
an analytic expression as a series in the temperature parameter $T$. 

The Mie potential contains two parameters $n, \, m$, restricted to $n>m>3$. If the ratio $m/n$ is a rational number, then $B(T)$ is a finite 
sum of hypergeometric functions. The case $n=2m$ is discussed in detail, providing asymptotic behaviors as $T \rightarrow \infty$ and 
$T \rightarrow 0$, this includes as special case the Lennard-Jones potential $(n=12, \, m=6$). The second special case $n=9$ and $m=6$ is also discussed and new results are obtained in the case $T \rightarrow 0$.

 In comparison with the evaluation by classical analytic procedures, the method of brackets produces a direct and 
simpler evaluation of the second virial coefficient. 

\section{Acknowledgments}

The first author wishes to thank the hospitality of the Mathematics Department where part of this work was 
developed.  The second author was supported in part by Fondecyt (Chile) Grants Nos. 1040368, 1050512 and 1121030, by DIUBB (Chile) Grant Nos. 102609,  
GI 153209/C  and GI 152606/VC.

\end{document}